\documentclass[conference]{IEEEtran}
\IEEEoverridecommandlockouts

\usepackage{color}

\usepackage{graphicx,tabularx,array,amsmath,amsthm,thmtools}

\usepackage{mathtools}
\DeclarePairedDelimiter{\ceil}{\lceil}{\rceil}

\usepackage{caption}
\usepackage{amsfonts}
\usepackage{bbm}
\usepackage{makecell}
\usepackage{multirow}
 \usepackage{amssymb}
\usepackage{txfonts}
\usepackage[T1]{fontenc}
\usepackage{tikz}
\usepackage[scr=dutchcal]{mathalfa}
\let\mathscr\mathbscr

\newtheorem{Theorem}{Theorem}
\newtheorem{Proposition}{Proposition}
\newtheorem{Lemma}{Lemma}

\newtheorem{Example}{Example}
\newtheorem{Remark}{Remark}
\newtheorem{Definition}{Definition}

\newcommand{\indep}{\rotatebox[origin=c]{90}{$\models$}}
\newtheorem{Claim}{Claim}

\ifCLASSINFOpdf
\else
\fi
\begin{document}
%
\title{On the Sub-optimality of Single-letter Coding in Multi-terminal Communications}

\author{\IEEEauthorblockN{Farhad Shirani Chaharsooghi}
\IEEEauthorblockA{Electrical Engineering and\\Computer Science\\
University of Michigan\\
Ann Arbor, Michigan, 48105\\
Email: fshirani@umich.edu }
\and

\IEEEauthorblockN{S. Sandeep Pradhan}
\IEEEauthorblockA{Electrical Engineering and\\Computer Science\\
University of Michigan\\
Ann Arbor, Michigan, 48105\\
Email: pradhanv@umich.edu}}


%


\maketitle

\begin{abstract}
We investigate binary block-codes (BBC). A BBC is defined as a vector of Boolean functions. We consider BBCs which are generated randomly, and using single-letter distributions. We characterize the vector of dependency spectrums of these BBCs. We use this vector to upper-bound the correlation between the outputs of two distributed BBCs. Finally, the upper-bound is used to show that the large blocklength single-letter coding schemes in the literature are sub-optimal in some multiterminal communication settings.
\end{abstract}


%
\IEEEpeerreviewmaketitle

\section{Introduction}

In his paper, "A Mathematical Theory of Communications" \cite{Shannon} - often regarded as the Magna Carta of digital communications - Shannon pointed out that in order to exploit the redundancy of the source in data compression, it is necessary to compress large blocks of the source simultaneously. More precisely, optimality is only approached as the effective length of the encoding functions approaches infinity. In this context, the effective length of an encoding function can be interpreted as the average number of input elements necessary to estimate an output element of the encoding function with high accuracy.
The same observation was made in the case of point-to-point (PtP) channel coding. As a result, a common feature of the coding schemes used in PtP communication is that they have large effective lengths. Loosely speaking, this means that each output element in these schemes is a function of the entire input sequence, where the length of the input sequence is asymptotically large. In the source coding problem, by compressing large blocks at the same time, one can exploit the redundancy in the source. In the channel coding problem, transmitting the input message over large blocks allows the decoder to exploit the typicality of the noise vector. These large blocklength gains due to the typicality of the noise vector in channel coding, and the source vector in source coding are also present in coding over networks. However, in multiterminal communication it is often desirable to maintain correlation among the output sequences at different nodes. This requirement can be due to explicit constraints in the problem statement such as joint distortion measures in multi-terminal source coding, or it can be due to implicit factors such as the need for interference reduction in multi-terminal channel coding, or it can be due to the nature of the shared communication channel. In the latter case, correlation between the outputs is necessary as a means for further cooperation among the transmitters. In this paper, we show that pairs of encoding functions with large effective lengths are inefficient in coordinating their outputs. This is due to the fact that such encoding functions are unable to produce highly correlated outputs from highly correlated inputs. The loss of correlation undermines the encoders' ability to cooperate and take advantage of the multi-terminal nature of the problem. In PtP communication problems, where there is only one transmitter, the necessity for cooperation does not manifest itself. For this reason, although encoders with asymptotically large effective lengths are optimal in PtP communications, they are sub-optimal in the network communication case. It is worth mentioning that there is a subtle difference between the effective length mentioned here and the blocklength of an encoding function. The effective length of an encoding function can be interpreted as the average number of input elements necessary to estimate an output element with high precision. Whereas the blocklength is the length of the sequence which is the input to the encoding function. It is well-known that the performance of block-codes is super-additive with respect to blocklength, meaning that the best performance of block-codes of a certain length is an increasing function of the blocklength. However, the performance as a function of the effective length of the code is not necessarily super-additive. 

The loss in correlation caused by the application of large effective-length codes causes a discontinuity in the performance of such codes in some multi-terminal problems. This was first observed for the problem of distributed source coding \cite{wagner}, in the Berger-Tung coding strategy. It was noted that when common information is available to the two encoders in the distributed source coding problem, the performance is discontinuously better than when the common information is replaced with highly correlated components. In \cite{FinLen}, we argued that the discontinuity in performance is due to the fact that the encoding functions in the Berger-Tung scheme preserve common information, but are unable to preserve correlation between highly correlated components. We proposed a new coding scheme, and derived an improved achievable rate-distortion region for the two user distributed source coding problem \cite{FinRD}. The new strategy uses a concatenated coding scheme which consists of one layer of codes with finite effective length, and one layer of codes with asymptotically large effective lengths.

In this paper, we prove the assertion that there is a discontinuity in the correlation preserving abilities of encoding functions produced using single-letter coding schemes such as the Berger-Tung scheme \cite{Markov}. The single-letter coding schemes considered in this work are general and include Shannon's point-to-point source coding scheme \cite{Shannon}, the Berger-Tung scheme \cite{Markov}, the Zhang-Berger scheme \cite{ZB}. The proof involves several steps. We make extensive use of the mathematical machinery developed in \cite{arxiv1}. 

In \cite{arxiv1}, we provide a bound on the correlation between the outputs of two arbitrary Boolean functions. The bound is presented as a function of the dependency spectrum of the Boolean functions. The dependency spectrum is a generalization of the effective length of a Boolean function, and is explained in more detail in the next sections. Here, we consider two arbitrary binary block codes (BBC) as defined in \cite{ComInf1}. The two encoding functions are applied to two correlated discrete, memoryless sources  (DMS). We define the correlation between the outputs of these encoding functions as the average probability that any two output-bits are equal, where the average is over the elements of the output vector. Using the bound in \cite{arxiv1}, we show that codes generated by large blocklength single-letter coding schemes are incapable of producing highly correlated outputs from highly correlated inputs. More precisely, we show that as the block-length increases, the outputs of the quantizers at each terminal become less correlated. This leads us to conjecture that such schemes are sub-optimal in network communication problems. We provide one example in the problem of lossless transmission of correlated sources over the interference channel, and prove that in this example single-letter coding schemes are sub-optimal. 

The rest of the paper is organized as follows: Section \ref{sec:Not} gives the notation used in this work. Section \ref{sec:sum} includes a summary of the mathematical tools in \cite{arxiv1} which are used here. In Section \ref{sec:Cor}, we derive a probabilistic bound on the correlation between the outputs of encoding functions produced using single-letter coding schemes. Section \ref{sec:Ex} includes an example where we show the sub-optimality of single-letter schemes in a specific multi-termianal problem. Section \ref{sec:conc} concludes the paper.
\section{Notation}\label{sec:Not}
In this section, we introduce the notation used in the paper. We represent random variables by capital letters such as $X, U$. Sets are denoted by calligraphic letters such as $\mathcal{X}, \mathcal{U}$. For a random variable $X$, the corresponding probability space is $(\mathcal{X}, \mathbf{F}_{X}, P_X)$, where $\mathbf{F}$ is the underlying $\sigma$-field. The set of all subsets of $\mathcal{X}$ is written as $2^{\mathcal{X}}$. There are three different notations used for different classes of vectors. For random variables, the $n$-length vector $(X_1,X_2,\cdots,X_n)$, where $X_i\in \mathcal{X}$, is denoted by $X^n\in \mathcal{X}^n$. For the vector of functions $(e_1(X),e_2(X),\cdots, e_n(X))$ we use the notation $\underline{e}(X)$. The binary string $(i_1,i_2,\cdots,i_n), i_j\in \{0,1\}$ is written as $\mathbf{i}$.  As an example, the set of functions $\{\underline{e}_{\mathbf{i}}(X^n)| \mathbf{i}\in \{0,1\}^n\}$ is the set of $n$-length vectors of functions $(e_{1,\mathbf{i}},e_{2,\mathbf{i}},\cdots,e_{n,\mathbf{i}})$ operating on the vector $(X_1,X_2,\cdots,X_n)$ each indexed by an $n$-length binary string $(i_1,i_2,\cdots,i_n)$. The vector of binary strings $(\mathbf{i}_1,\mathbf{i}_2,\cdots, \mathbf{i}_n)$ denotes the standard basis for the $n$-dimensional space (e.g. $\mathbf{i}_1=(0,0,\cdots,0,1)$). The vector of random variables $(X_{j_1},X_{j_2},\cdots, X_{j_k}), j_i\in [1,n], j_i\neq j_k$, is denoted by $X_{\mathbf{i}}$, where $i_{j_l}=1, \forall l\in [1,k]$. For example, take $n=3$, the vector $(X_1,X_3)$ is denoted by $X_{101}$, and the vector $(X_1,X_2)$ by $X_{110}$.  Particularly, $X_{\mathbf{i}_j}=X_j , j\in [1,n]$. Also, for $\mathbf{t}=\underline{1}$, the all-ones vector, $X_{\mathbf{t}}=X^n$.
For two binary strings $\mathbf{i},\mathbf{j}$, we write $\mathbf{i}<\mathbf{j}$ if and only if $i_k<j_k, \forall k\in[1,n]$. For a binary string $\mathbf{i}$ we define $N_{\mathbf{i}}\triangleq w_H(\mathbf{i})$, where $w_H$ denotes the Hamming weight. Lastly, the vector $\sim \mathbf{i}$ is the element-wise complement of $\mathbf{i}$.
\section{Correlation and the Dependency Spectrum}\label{sec:sum}
Recently, we introduced the \textit{`effective length'} of an additive Boolean function, and its extension the \textit{`Dependency spectrum'} of a general Boolean function \cite{arxiv1}. We used the dependency spectrum to bound the correlation between the outputs of functions of sequences of random variables. 
Loosely speaking, for the Boolean function $f:\{0,1\}^n\to \{0,1\}$, the dependency spectrum is a vector which characterizes the correlation between the output $f(X^n)$ with any of the subsequences of the input sequence. In this work, we make extensive use of the upper-bound in \cite{arxiv1}, as well as the mathematical apparatus developed there. Hence, we provide a concise review of the material.

Let $e:\{0,1\}^n\to \{0,1\}$ be a Boolean function. Following the method presented in \cite{ComInf2}, we convert the problem of analyzing a Boolean function into one where the encoder is a binary real-valued function. Mapping the discrete-valued function to a real-valued one is crucial since it allows us to use the rich set of tools available in functional analysis. Assume that $P(e(X^n)=1)=q$. The corresponding real function for $e$ is defined as:
\begin{align}
\tilde{e}(X^n)=    \begin{cases}
      1-q, & \qquad  e(X^n)=1, \\
      -q. & \qquad\text{otherwise}.
    \end{cases}
\end{align}
The following gives the definition for the additive decomposition of $\tilde{e}$.
\begin{Definition}
 The vector of real functions $(\tilde{e}_{\mathbf{i}})_{\mathbf{i}\in\{0,1\}^n}$ is called the additive decomposition of $\tilde{e}$, where $\tilde{e}_{\mathbf{i}}=\mathbb{E}_{X^n|X_{\mathbf{i}}}(\tilde{e}|X_{\mathbf{i}})-\sum_{\mathbf{j}< \mathbf{i}} \tilde{e}_{\mathbf{j}}$.
\end{Definition}
\begin{Remark}
 From the above definition, it follows that $\tilde{e}=\sum_{\mathbf{i}\in \{0,1\}^n}\tilde{e}_{\mathbf{i}}$.
\end{Remark}
The following proposition describes the properties of the components of the additive decomposition:
\begin{Proposition}
\label{prop:belong2}
Define $\mathbf{P}_{\mathbf{i}}$ as the variance of $\tilde{e}_{\mathbf{i}}$. The following hold:
\\ 1) $\forall \mathbf{i}\leq \mathbf{k}$, we have $\mathbb{E}_{X^n|X_{\mathbf{j}}}(\tilde{e}_{\mathbf{i}}|X_{\mathbf{k}})=\tilde{e}_{\mathbf{i}}$.\\
 2) $\mathbb{E}_{X^n}(\tilde{e}_{\mathbf{i}}\tilde{e}_{\mathbf{k}})=0$, for $\mathbf{i}\neq \mathbf{k}$.\\
 3) $\forall \mathbf{k}\leq \mathbf{i}: \mathbb{E}_{X^n|X_{\mathbf{k}}}(\tilde{e}_{\mathbf{i}}|X_{\mathbf{k}})=0.$
 \\ 4)  $\mathbf{P}_{\mathbf{i}}=\mathbb{E}(\tilde{e}_{\mathbf{i}}^2)=c_{\mathbf{i}}^2(q(1-q))^{w_{H}(\mathbf{i})}.$
\end{Proposition}
The following lemma gives the formula for calculating $\mathbf{P}_{\mathbf{i}}$.
\begin{Lemma}
For arbitrary $e:\{0,1\}^n\to \{0,1\}$, let $\tilde{e}$ be the corresponding real function. The variance of each component $\tilde{e}_{\mathbf{i}}$ in the additive decomposition is given by the following recursive formula $\mathbf{P}_{\mathbf{i}} =Var(\tilde{e}_{\mathbf{i}})= \mathbb{E}_{X_{\mathbf{i}}}(\mathbb{E}_{X^n|X_{\mathbf{i}}}^2(\tilde{e}|X_{\mathbf{i}}))-\sum_{\mathbf{j}< \mathbf{i}}\mathbf{P}_{\mathbf{j}}, \forall \mathbf{i}\in \mathbb{F}_2^n$, where $\mathbf{P}_{\underline{0}}\triangleq 0$. 
\label{Lem:power}
\end{Lemma}
The dependency spectrum of the function is defined below.
\begin{Definition}[Dependency Spectrum]
 For a Boolean function $e$, the vector of variances $(P_{\mathbf{i}})_{\mathbf{i}\in \{0,1\}^n}$ is called the dependency spectrum of $e$.
\end{Definition}
 Let $(X,Y)$ be a pair of DMS's. Consider two arbitrary Boolean functions $e:\mathcal{X}^n\to \{0,1\}$ and $f:\mathcal{Y}^n\to \{0,1\}$. Let  $ q\triangleq P(e(X^n)=1)$, $r\triangleq P(f(Y^n)=1)$, and $\epsilon\triangleq P(X\neq Y)$. Also, let $\tilde{e}=\sum_{\mathbf{i}}e_{\mathbf{i}}$, and $\tilde{f}=\sum_{\mathbf{i}}f_{\mathbf{i}}$ be the additive decomposition of these functions. 
\begin{Theorem}
The following bound holds:
 \begin{align*}
&2\!\!\sqrt{\sum_{\mathbf{i}}\mathbf{P}_{\mathbf{i}}}\sqrt{\sum_{\mathbf{i}}\mathbf{Q}_{\mathbf{i}}}-2\!\!\sum_{\mathbf{i}}C_\mathbf{i}\mathbf{P}_{\mathbf{i}}^{\frac{1}{2}}\mathbf{Q}_{\mathbf{i}}^{\frac{1}{2}} 
\leq  \!\!P(e(X^n)\neq f(Y^n))
\\&\leq 1- 2\sqrt{\sum_{\mathbf{i}}\mathbf{P}_{\mathbf{i}}}\sqrt{\sum_{\mathbf{i}}\mathbf{Q}_{\mathbf{i}}}+2\sum_{\mathbf{i}}C_\mathbf{i}\mathbf{P}_{\mathbf{i}}^{\frac{1}{2}}\mathbf{Q}_{\mathbf{i}}^{\frac{1}{2}} 
,
\end{align*}
 where $C_{\mathbf{i}}\triangleq  (1-2\epsilon)^{N_\mathbf{i}}$,  and ${\tilde{e}}$ ($\tilde{f}$) is the real function corresponding to ${e}$ ($f$), and $\mathbf{P}_{\mathbf{i}}$ ($\mathbf{Q}_{\mathbf{i}}$) is the variance of $\tilde{e}_{\mathbf{i}}$ ($\tilde{f}_{\mathbf{i}}$), and finally, $N_{\mathbf{i}}\triangleq w_H(\mathbf{i})$.
 \label{th:sec3}
\end{Theorem}
\section{Correlation Preservation in Single-letter BBCs}\label{sec:Cor}
In this section, we provide the definition of a `coding scheme', and identify three properties which are shared among many of the coding schemes available for multi-terminal communications. We call the group of coding schemes which share these properties the Single-letter Random Coding Schemes (SLCS). These coding schemes include Shannon's PtP source coding scheme, the Berger-Tung coding scheme for distributed source coding \cite{Markov}, and the Zhang-Berger multiple-descriptions coding scheme \cite{ZB}. We use the results in the previous section to bound the correlation between the outputs of two encoding functions produced using a SLCS. The proof involves several steps. First, it is shown that SLCSs produce encoding functions which have most of their variance either on the single-letter components of their dependency spectrum or on the components with asymptotically large blocklengths. This along with Theorem \ref{th:sec3} are used to prove that such schemes are inefficient in preserving correlation. The following defines a coding scheme:
\begin{Definition}
 A Coding Scheme $\mathscr{S}$ is characterized by a probability measure $P_{\mathscr{S}}(\underline{e}_1,\underline{e}_2,\cdots,\underline{e}_t)P_{\mathscr{S}}(\underline{d}_1,\underline{d}_2,\cdots,\underline{d}_r|\underline{e}_1,\underline{e}_2,\cdots,\underline{e}_t)$ on the set of encoding  functions $\underline{e}_k, k\in [1,t]$, and decoding functions $\underline{d}_k, k\in [1,r]$.
\end{Definition}
For example, Shannon's point-to-point coding schemes determine $P_{\mathscr{S}}(\underline{e})P_{\mathscr{S}}(\underline{d}|\underline{e})$ by using a single-letter distribution to assign probabilities to codebooks, and then using typicality encoding and decoding rules to determine the encoding and decoding functions.  Whenever the choice of the coding scheme is clear, we denote the distribution by $P_{\underline{E}}(\underline{e})P_{\underline{D}|\underline{E}}(\underline{d}|\underline{e})$. The following defines the SLCSs:
\begin{Definition} \label{def:SLCS}
 The coding scheme characterized by $P_{\mathscr{S}}$ is called an SLCS if the following constraints are satisfied. For an arbitrary $k\in [1,t]$, let $\underline{E}=\underline{E}_k$, then

 1) $\forall x^n, \exists \delta_X>0, B_n(x^n)\subset \mathcal{X}^n$ { such that } \\$P_{X^n}(X^n\in B_n(x^n))\leq 2^{-n\delta_{X}}$, and $\forall y^n\notin B_n(x^n)$: $ \underline{E}(x^n) \indep \underline{E}(y^n)$ .
 
 
 2) $ \forall \delta>0, \exists n\in \mathbb{N}\ni \forall m>n, \forall x^m\in \{0,1\}^m , v \in \{0,1\}, \forall i\in [1,m]$:
 \\ $|P_{\mathscr{S}}(E_i(X^m)=v|X^m=x^m)-P_{\mathscr{S}}(E_i(X^m)=v|X_i=x_i)|<\epsilon$.

3) $\forall \pi \in S_n: P_{\underline{E}}(\underline{E})=P_{\underline{E}}(\underline{E}_\pi)$, where $\underline{E}_\pi(X^n)=\pi^{-1}(\underline{E}(\pi(X^n)))$, where $S_n$ is the symmetric group of length n.

\end{Definition}

\begin{Remark}
 The first condition can be interpreted as follows: for an arbitrary sequence $x^n$, let $B_n(x^n)$ be the set of sequences $y^n$, such that the set of encoding functions which map $x^n$ and $y^n$ to the same sequence has non-zero probability with respect to the measure $P_{\mathscr{S}} $ (e.g. in typicality encoding this requires $x^n$ and $y^n$ to be jointly typical based on some fixed distribution). Then, the condition requires that the probability of the set $B_n(x^n)$ goes to $0$ exponentially as $n\to \infty$. Furthermore, if $y^n\notin B_n(x^n)$, then it is mapped to a codeword which is chosen independent of $\underline{E}(x^n)$ (i.e. codewords are chosen pairwise independently). 
 \end{Remark}
 \begin{Remark}
The second condition can be interpreted as follows: the joint distribution of the input sequence and the output sequence of the encoding function averaged over all possible encoding functions approaches a product distribution in variational distance as $n\to \infty$.
\end{Remark}
\begin{Remark}
The explanation for the third condition is that the probability that a vector $x^n$ is mapped to $y^n$ depends only on their joint type and is equal to the probability that the permuted sequence $\pi(x^n)$ is mapped to $\pi(y^n)$. As an example, this condition holds in typicality encoding. 
\end{Remark}
In the next example we show that the Shannon's coding scheme for point-to-point source coding satisfy the above conditions.

\begin{Example}
 \begin{figure}[!t]
\centering
\includegraphics[width=\columnwidth]{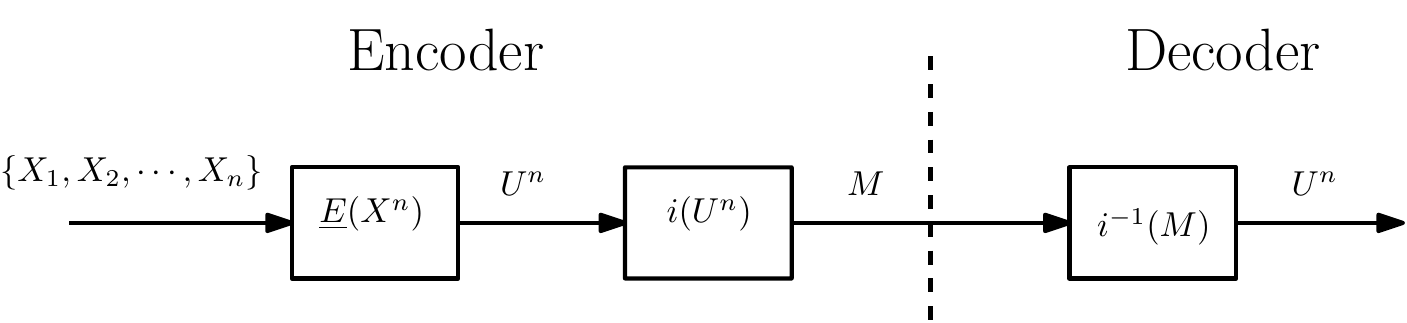}
\caption{Point-to-point source coding example}
\label{fig:PtPS}
\end{figure}
Consider the PtP source coding problem depicted in Figure \ref{fig:PtPS}. A discrete memoryless source $X$ is being fed to an encoder. The encoder utilizes the mapping $\underline{E}:\mathcal{X}^n\to \mathcal{U}^n$ to compress the source sequence. The image of $\underline{E}$ is indexed by the bijection $i:\textit{Im}(\underline{E})\to [1,|\textit{Im}(\underline{Q})|]$. The index $M\triangleq{i(E(X^n))}$ is sent to the decoder. The decoder reconstructs the compressed sequence $U^n\triangleq i^{-1}(M)=\underline{E}(X^n)$.  The efficiency of the reconstruction is evaluated based on the separable distortion criteria $d_n:\mathcal{X}^n\times\mathcal{U}^n\to [0,\infty)$. The separability property means that $d_n(x^n,u^n)=\sum_{i\in[1,n]}d_1(x_i,u_i)$. We assume that the alphabets $\mathcal{X}$ and $\mathcal{U}$ are both binary. The rate of transmission is defined as $R\triangleq \frac{1}{n}\log{|\textit{Im}(\underline{E})|}$, and the average distortion is defined as $\frac{1}{n}\mathbb{E}(d_n(X^n, U^n))$. The goal is to choose $\underline{E}$ such that the rate-distortion tradeoff is optimized. Note that the choice of the bijection `$i$' is irrelevant to the performance of the system. The following Lemma gives the achievable RD region for this setup.  
\begin{Lemma}\cite{Shannon}
 For the source $X$ and distortion criteria $d_1:\{0,1\}\times\{0,1\}\to [0,\infty)$, fix a conditional distribution  $p_{U|X}(u|x),x,u\in \{0,1\}$. 
 The rate-distortion pair $(R,D)=\left(r, \mathbb{E}_{X,U}(d_1(X,U)\right)) $ is achievable for all $r<I(U;X)$. 
\end{Lemma}
\begin{proof}
In order to verify the properties of the SLCS's in the coding scheme used in this problem, we give an outline of the scheme. Fix $n\in \mathbb{N}$, and $\epsilon>0$. Define $P_U(u)=\mathbb{E}_{X}\{P_{U|X}(u|X))\}$. Proving achievability is equivalent to showing the existence of a suitable encoding function $\underline{E}(X^n)$. In \cite{Shannon}, a randomly generated encoding function is constructed with the aid of a set of vectors called the codebook, and an assignment rule called typicality encoding. We construct the codebook $\mathcal{C}$ as follows.  Let ${A}_{\epsilon}^n(U)\triangleq \{u^n\big||\frac{1}{n}w_{H}(u^n)-P_{U}(1)|<\epsilon\}$ be the set of $n$-length binary vectors which are $\epsilon$-typical with respect to $P_U$. Choose $\ceil{2^{nR}}$ vectors from ${A}_{\epsilon}^n(U)$ randomly and uniformly. Let $\mathcal{C}\subset {A}_{\epsilon}^n(U)$ be the set of these vectors. The encoder constructs the encoding function $\underline{E}(X^n)$ as follows.  For an arbitrary sequence $x^n\in \{0,1\}^n$, define $A_{\epsilon}^n(U|x^n)$ as the set of vectors in $\mathcal{C}$ which are jointly $\epsilon$-typical with $x^n$ based on $P_{U|X}$. The vector $\underline{E}(x^n)$ is chosen randomly and uniformly from $A_{\epsilon}^n(U|x^n)\cap \mathcal{C}$. The probabilistic choice of the codewords as well as the quantization, puts a distribution on the random function $\underline{E}$. It can be shown that as $n$ becomes larger codes produced based on this distribution $P(\underline{Q})$ achieve the rate-distortion vector (R,D) with probability approaching one. 
 \end{proof}
 \label{Ex:PtPS}
\end{Example}
\begin{Remark}
 It is well-known that in the above scheme, the codebook generation process could be altered in the following way. Instead of choosing the codewords randomly and uniformly from the set of typical sequences ${A}_{\epsilon}^n(U)$, the encoder can produce each codeword independent of the others and with the distribution $P_{U^n}(u^n)=\Pi_{i\in[1,n]}P_{U}(u_i)$. However, the discussion that follows remains unchanged regardless of which of these codebook generation methods are used.  
\end{Remark}

1) Codewords are chosen pairwise independently. So,  for two input sequences $x^n$, and $y^n$ given that $A_{\epsilon}^n(U|x^n)\cap A_{\epsilon}^n(V|y^n)=\phi$, the two vectors are not mapped to the same codeword, hence they are mapped to independently generated codewords (i.e. $\underline{E}(x^n)$ is chosen independently of $\underline{E}(y^n)$).


2) As $n$ becomes large, the $i$th output element $E_i(X^n)$ is correlated with the input sequence $X^n$ only through the $i$th input element $X_i$:
\begin{align*}
 &\forall \delta>0, \exists n\in \mathbb{N}: m>n
 \Rightarrow \forall x^m\in \{0,1\}^m , v \in \{0,1\},
 \\& |P_{\mathscr{S}}(Q_i(X^m)=v|X^m=x^m)-P_{\mathscr{S}}(Q_i(X^m)=v|X_i=x_i)|<\delta. 
\end{align*}
\begin{proof}
 For a fixed quantization function $\underline{e}:\{0,1\}^m\to\{0,1\}^m$, $\underline{e}(X^m)$ is a function of $X^m$. However,
 without the knowledge that which encoding function is used, $E_i(X^m)$ is related to $X^m$ only through $X_i$. In other words, averaged over all encoding functions, the effects of the rest of the elements diminishes. We provide a proof of this statement below:
 
First, we are required to provide some definitions relating to the joint type of pairs of sequences. For binary strings $u^m, x^m$, define $N(a,b|u^m,x^m)\triangleq
|\{j|u_j=a, x_j=b\}|$, that is the number of indices $j$ for which the value of the pair $(u_j,x_j)$ is $(a,b)$. For $s,t\in \{0,1\}$, define $l_{s,t}\triangleq N(s,t|u^m,x^m)$, the vector $(l_{0,0}, l_{0,1}, l_{1,0}, l_{1,1})$ is called the joint type of $(u^m,x^m)$. For fixed $x^m$ The set of sequences $T(l_{0,0}, l_{0,1}, l_{1,0}, l_{1,1})=\{u^m| N(s,t|u^m,x^m)=l_{s,t}, s,t\in\{0,1\}\}$, is the set of vectors which have joint type $(l_{0,0}, l_{0,1}, l_{1,0}, l_{1,1})$ with the sequence $x^m$. Fix $m, \epsilon>0$, and define $\mathcal{L}_{\epsilon,n}\triangleq \{(l_{0,0}, l_{0,1}, l_{1,0}, l_{1,1})|  |\frac{l_{s,t}}{m}-P_{U,X}(s,t)|<\epsilon\}$. Then for the conditional typical set $A_{\epsilon}^n(U|x^m)$ defined above we can write
\begin{align*}
 A_{\epsilon}^n(U|x^m)=\cup_{(l_{0,0}, l_{0,1}, l_{1,0}, l_{1,1})\in \mathcal{L}_{\epsilon,n}} T(l_{0,0}, l_{0,1}, l_{1,0}, l_{1,1}).
\end{align*}
The type of $x^m$, denoted by $(l_0,l_1)$ is defined in a similar manner.
 Since $E_i(X^m)$ are chosen uniformly from the set $A_{\epsilon}^n(U|x^m)$, we have:
\begin{align*}
&P_{\mathscr{S}}(E_i(X^m)=v|X^m=x^m)=\frac{|\{u^m| u_1=v, u^m\in  A_{\epsilon}^m(U|x^m)\}|}{|\{u^m| u^m\in  A_{\epsilon}^m(U|x^m)\}|} 
\\&=\frac{\sum_{(l_{0,0}, l_{0,1}, l_{1,0}, l_{1,1})\in  \mathcal{L}_{\epsilon,n}} |\{u^m| u_1=v, u^m\in  T(l_{0,0}, l_{0,1}, l_{1,0}, l_{1,1})\}|}{\sum_{(l_{0,0}, l_{0,1}, l_{1,0}, l_{1,1})\in  \mathcal{L}_{\epsilon,n}} |\{u^m|  u^m\in  T(l_{0,0}, l_{0,1}, l_{1,0}, l_{1,1})\}|}\\
&=\frac{\sum_{(l_{0,0}, l_{0,1}, l_{1,0}, l_{1,1})\in  \mathcal{L}_{\epsilon,n}} {{l_{x_1}-1} \choose {l_{u_1,x_1}-1} }{{l_{\bar{x}_1}} \choose {l_{u_1,\bar{x}_1}} }}{\sum_{(l_{0,0}, l_{0,1}, l_{1,0}, l_{1,1})\in  \mathcal{L}_{\epsilon,n}}{{l_{x_1}} \choose {l_{u_1,x_1}} }{{l_{\bar{x}_1}} \choose {l_{u_1,\bar{x}_1}} }}
\\&=\frac{\sum_{(l_{0,0}, l_{0,1}, l_{1,0}, l_{1,1})\in  \mathcal{L}_{\epsilon,n}} \frac{(l_{x_1}-1)!}{
(l_{u_1,x_1}-1)!(l_{x_1}-l_{u_1,x_1})!}
 \frac{l_{\bar{x}_1}!}{
l_{u_1,\bar{x}_1}!(l_{\bar{x}_1}-l_{u_1,\bar{x}_1})!}}
{\sum_{(l_{0,0}, l_{0,1}, l_{1,0}, l_{1,1})\in  \mathcal{L}_{\epsilon,n}} \frac{l_{x_1}!}{
l_{u_1,x_1}!(l_{x_1}-l_{u_1,x_1})!}
 \frac{l_{\bar{x}_1}!}{
l_{u_1,\bar{x}_1}!(l_{\bar{x}_1}-l_{u_1,\bar{x}_1})!}}
\\&\stackrel{(a)}{=}\frac{\sum_{(l_{0,0}, l_{0,1}, l_{1,0}, l_{1,1})\in  \mathcal{L}_{\epsilon,n}} l_{u_1,x_1} \frac{1}{
l_{u_1,x_1}!(l_{x_1}-l_{u_1,x_1})!}
 \frac{1}{
l_{u_1,\bar{x}_1}!(l_{\bar{x}_1}-l_{u_1,\bar{x}_1})!}}
{l_{x_1}\sum_{(l_{0,0}, l_{0,1}, l_{1,0}, l_{1,1})\in  \mathcal{L}_{\epsilon,n}} \frac{1}{
l_{u_1,x_1}!(l_{x_1}-l_{u_1,x_1})!}
 \frac{1}{
l_{u_1,\bar{x}_1}!(l_{\bar{x}_1}-l_{u_1,\bar{x}_1})!}}
\\&\stackrel{(b)}{\Rightarrow} \frac{P_{U,X}(u_1,x_1)-\epsilon}{P_X(x_1)+\epsilon} \leq P_{\mathscr{S}}(Q_i(X^m)=v|X^m=x^m) \leq  \frac{P_{U,X}(u_1,x_1)+\epsilon}{P_X(x_1)-\epsilon}
\\& \Rightarrow \exists m,\epsilon>0: |P_{\mathscr{S}}(Q_i(X^m)=v|X^m=x^m)-P_{U|X}(u_1|x_1)|\leq \delta.
\end{align*}
In (a), we use the fact that for fixed $x^m$, $(l_{x_1},l_{\bar{x}_1})$ is fixed to simplify the numerators. In (b) we have used that for jointly typical $\epsilon$-sequences $(u^m,x^m)$, $l_{u_1,x_1}\in [n(P_{U,X}(u_1,x_1)-\epsilon), n(P_{U,X}(u_1,x_1)+\epsilon)]$, and $l_{x_1}\in  [n(P_{X}(x_1)-\epsilon), n(P_{X}(x_1)+\epsilon)]$.

\end{proof}

3) The encoder is insensitive to permutations. Due to typicality encoding the probability that a vector $x^n$ is mapped to $y^n$ depends only on their joint type and is equal to the probability that $\pi(x^n)$ is mapped to $\pi(y^n)$.

Our goal is to analyze the correlation preserving properties of SLCS's. More precisely, we investigate two encoding functions generated using a SLCS, and bound the correlation between the outputs of these two functions. For a randomly generated encoding function $\underline{E}= (E_1,E_2,\cdots,E_n)$, denote the additive decomposition of the real function corresponding to the $k$th element as $\tilde{E}_k= \sum_{\mathbf{i}}\tilde{E}_{k, \mathbf{i}}, k\in [1,n]$. Let $\mathbf{P}_{j,\mathbf{i}}$ be the variance of $\tilde{E}_{k, \mathbf{i}}$. The next theorem shows that SLCSs produce encoding functions which have most of their variance either on the single-letter components of their dependency spectrum or on the components with asymptotically large blocklengths. 
\begin{Theorem}
 For any $k\in [1,n], m\in \mathbb{N}, \gamma>0$, $P_{\mathscr{S}}(\sum_{{\mathbf{i}}:N_{\mathbf{i}}\leq m, \mathbf{i}\neq \mathbf{i}_k} \mathbf{P}_{k,{\mathbf{i}}}\geq \gamma)\to 0$, as $n\to \infty$. Where, $\mathbf{i}_k$ is the $k$th standard basis element.
 \label{th:sec4}
\end{Theorem}
\begin{IEEEproof}
 Please refer to the Appendix.
\end{IEEEproof}

 Theorem \ref{th:sec4} shows that SLCS's distribute most of the variance of $\tilde{E}_k$ on $ \tilde{E}_{k, \mathbf{i}}$'s  which operate on asymptotically large blocks of the input, and on the single-letter component $\tilde{E}_{k, \mathbf{i}_k}$. Hence, the encoders generated using such schemes have high expected variance for decomposition elements with large effective lengths. This along with Theorem \ref{th:sec3} gives an upper bound on the correlation preserving properties of SLCS's. The following theorem states this upper bound which is the main result of this section:
\begin{Theorem}
Let $(X,Y)$ be a pair of DMS's, with $P(X=Y)=1-\epsilon, \epsilon>0$. Also, assume that the pair $\underline{E},\underline{F}:\{0,1\}^n\to \{0,1\}$ are produced using a SLCS characterized by $P_{\mathscr{S}}$. Define $E\triangleq E_{1}$, and $F\triangleq F_1$.
Then, $\forall \delta>0$:
\begin{align*}
 &P_{\mathscr{S}}\left(P_{X^n,Y^n}\left(E(X^n)\neq F(Y^n)\right)\!>\!\! 2\mathbf{P}^{\frac{1}{2}}\mathbf{Q}^{\frac{1}{2}}\!\!
 - \!\!2(1-2\epsilon)\mathbf{P}_{\mathbf{i}_1}^{\frac{1}{2}}\mathbf{Q}_{\mathbf{i}_1}^{\frac{1}{2}}-\delta\right)
\to 1,
\end{align*}
as $n\to \infty$. Where $\mathbf{P}_{\mathbf{i}}\triangleq Var(\tilde{E}_{\mathbf{i}})$, $\mathbf{Q}_{\mathbf{i}}\triangleq Var(\tilde{F}_{\mathbf{i}})$, $\mathbf{P}\triangleq Var(\tilde{E})$, and $ \mathbf{Q}\triangleq Var(\tilde{F})$.
    \label{th:main}
    \end{Theorem}
\begin{IEEEproof}
 Please refer to Appendix.
\end{IEEEproof}
\begin{Remark}
 The result in the theorem is also valid in the case when the SLCS sets $E=F$, and $E$ is produced based on $P_{\mathscr{S}}$ (i.e. when both encoders use the same encoding function.).
\end{Remark}
In order to increase the correlation between the outputs of encoding functions produced using SLCSs, the variance of $\tilde{E}_{k,\mathbf{i}_k}, k\in [1,n]$ needs to increase. This in turn increases the single-letter mutual information between the input and output of the encoder, which would require higher rates. For example consider the extreme case where $Var(\tilde{E}_{k,\mathbf{i}_k})$ is set to maximum (i.e. $Var(\tilde{E}_{k,\mathbf{i}_k})= Var(\tilde{E}_k)$). This requires ${E}_k(X^n)=X_k$. So, in order to achieve maximum correlation, the encoder must use uncoded transmission. 
 
 Theorem \ref{th:main} shows that there is a discontinuity in the correlation preserving ability of SLCSs as a function of the joint distribution $P(X,Y)$. 
Particularly, when $X=Y$, there is common-information \cite{ComInf1} available at the encoders. If the encoders use the same encoding function $\underline{E}(X^n)=\underline{E}(Y^n)$, their outputs would be equal with probability one. Whereas, if $P(X=Y)=1-\epsilon$, for any non-zero $\epsilon$, the output correlation is bounded away from 0 by the bound provided in Theorem \ref{th:main}, and the probability of equal outputs does not approach $1$ as $n\to \infty$. So, the correlation between the outputs of these encoding functions is discontinuous as a function of $\epsilon$. A special case of this phenomenon was observed in \cite{wagner}, in the distributed source coding problem, and was used to prove the sub-optimality of the Berger-Tung strategy.

\section{Multi-terminal Communication Examples}\label{sec:Ex}
In this section, we provide an example of a multi-terminal communication setup where SLCSs have suboptimal performance. We use the discontinuity mentioned in the previous section to show the sub-optimality of SLCSs.

Consider the problem of transmission of correlated sources over the interference channel (ICCS) described in \cite{ICcorr}. We examine the specific ICCS setup shown in Figure \ref{fig:ICcorr}. Here, the sources $X$ and $Y$ are Bernoulli random variables with parameters $\alpha_X$ and $\alpha_Y$, and $Z$ is a $q$-ary random variable with distribution $P_Z$. $X$ and $Z$ are independent. $Y$ and $Z$ are also independent. Finally, $X$ and $Y$ are correlated, and $P(X\neq Y)=\epsilon$. The random variable $N_{\delta}$ is Bernoulli with parameter $\delta$. The first transmitter transmits the binary input $X_1$, and the second transmitter transmits the pair of inputs $(X_{21}, X_{22})$, where $X_{21}$ is $q$-ary and $X_{22}$ is binary. Receiver 1 receives $Y_1=X_1\oplus_2N_{\delta}$, and receiver 2 receives $Y_2$ which is given below:
\begin{align}\label{eq:y2}
 Y_2= 
 \begin{cases}
 X_{21}, \qquad \qquad &\text{if }X_{22}=X_1,\\
 e,     & otherwise.
\end{cases}
\end{align}
\begin{figure}[!t]
\centering
\includegraphics[height=1.2in]{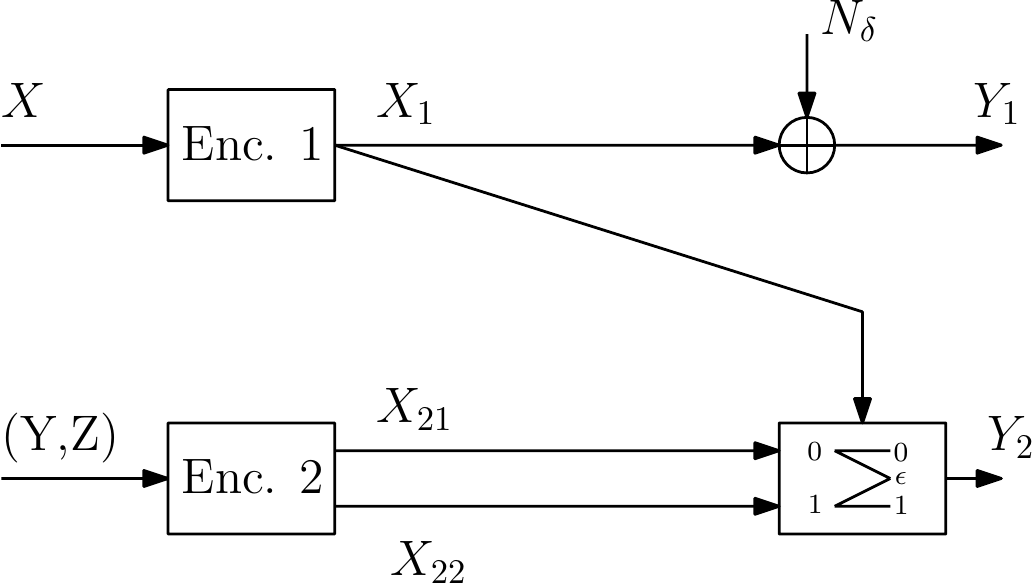}
 \caption{An ICCS example where SLCSs are suboptiomal.}
\label{fig:ICcorr}
\end{figure}
So, the second channel outputs $X_{21}$ noiselessly if the second encoder `guesses' the first encoder's output correctly (i.e. $X_{22}=X_1$), otherwise an erasure is produced. The following proposition gives a set of sufficient conditions for the transmission of correlated sources over this interference channel:
\begin{Proposition}\label{prop:ICcorrach}
 The sources $X$ and $Z$ are transmissible if there exist $\epsilon, \gamma,d >0$, and $n\in \mathbb{N}$ such that:
 \begin{align*}
 & H(X)\leq (1-h_b(\delta))\left(1- \frac{h_b(\gamma+d)}{1-h_b(\delta)}\right)+O\left(\frac{1+ \sqrt{nV+k\mathcal{V}(d)}Q^{-1}(\gamma)}{\sqrt{n}}\right),\\
 &  H(Z)\leq \left((1-\epsilon)^k\log{q}\right)\left(1- \frac{h_b(\gamma+d)}{1-h_b(\delta)}\right),
\end{align*}
 where $h_b(\cdot)$ is the binary entropy function, $V=\delta(1-\delta)log_2{(\frac{1-\delta}{\delta})}$ is the channel dispersion, and $\mathcal{V}(d)$ is the rate-dispersion function as in \cite{KostinaSC}, and $Q(.)$ is the Gaussian complementary cumulative distribution function.
  
 For a fixed $n$, $\epsilon$ and $\gamma$, we denote the set of pairs $(H(X), H(Z))$ which satisfy the bounds by $S(n,\epsilon,\gamma)$.
 \end{Proposition}
 \begin{IEEEproof}
First we provide an outline of the coding strategy. Fix $n,m\in \mathbb{N}, d, \gamma\in \mathbb{R}$, where $n\ll m$. Let $k=n\left(1-\frac{h_b(d+\gamma)}{1-h_b(\delta)}\right)^{-1}$. The encoders send $km$ bits of the compressed input at each block of transmission.
The first encoder transmits its source in two steps. First, it uses a fixed blocklength source-channel code \cite{KostinaSC} with parameters $(k,n,d,\gamma)$ . The code maps $k$-length blocks of the source to $n$-length blocks of the channel input, and the average distortion resulting from the code is less than $d+\gamma$. In this step, the encoder transmits the source in $m$ blocks of length $k$. A total of $nm$ channel uses are needed (note that n<k). In the second step, the encoder uses a large blocklength code to correct the errors in the previous step. 
The code has rate close to $\frac{h_b(\gamma+d)}{1-h_b(\delta)}$, and its input length is equal to $km$. 

The second encoder only transmits messages in the first step of transmission. It uses the same fixed blocklength code as the first encoder and the source sequence $Y^k$ to estimate the outcome of the first encoder. It sends this estimate of the first encoder's output on $X^n_{22}$. Since $P(X^k=Y^k)=(1-\epsilon)^k$, we conclude that  $X^k_1$ and $X^k_{22}$ are equal at least with probability $(1-\epsilon)^k$. The encoder sends the source $Z$ using $X_{21}$ over the resulting $q$-ary erasure channel which has probability of erasure at most $(1-\epsilon)^k$. The following provides a detailed descriptions of the coding strategy:
\\\textbf{Codebook Generation:} Fix $n, \epsilon, d$. Let $k=n\left(1-\frac{h_b(d+\gamma)}{1-h_b(\delta)}\right)^{-1}$. 
Let $C_k$ be the optimal source-channel code with parameters $(k,n,d,\gamma)$ for the point-to-point transmission of a binary source over the binary symmetric channel, as described in \cite{KostinaSC}. The code transmits $k$-length blocks of the source using $n$-length blocks of the channel input; and guarantees that the resulting distortion at each block is less than $d$ with probability $(1-\epsilon)$ (i.e $P(d_H(X^n,\hat{X}^n)>d)\leq \gamma$, where $\hat{X}$ is the reconstruction of the binary source $X$ at the decoder). In \cite{KostinaSC}, it is shown that the parameters of the code satisfy:
\begin{align*}
 n(1-h_b(\delta))-&k\left(H(X)-h_b(\alpha_X\ast d)\right)
 \\&= \sqrt{nV+k\mathcal{V}(d)}Q^{-1}(\gamma)+O({log(n)}).
\end{align*}
Since $P(d_H(X^n,\hat{X}^n)>d)\leq \gamma$, it is straightforward to show that the average distortion is less than or equal to $\gamma+d$. Also, construct a family of good channel codes $C'_m, m\in \mathbb{N}$ for the binary symmetric channel with rate $R_m=1-h_b(\delta)-\lambda_m$, where $\lambda_m\to 0$ as $m\to \infty$. Next, construct a family of good channel codes $C''_m, m\in \mathbb{N}$ for the $q$-ary erasure channel with rate $R_m=(1-\epsilon)^klog(q)-\lambda_m$. Finally, randomly and uniformly bin the space of binary vectors of length $kn$ with rate $R'=h_n(d+\gamma)$. More precisely, generate a binning function $B:\{0,1\}^{km}\to \{0,1\}^{kmR'}$, by mapping any vector $\mathbf{i}$ to a value chosen uniformly from $\{0,1\}^{kmR'}$.
\\\textbf{Encoding:} Fix $m$. At each block the encoders transmit $km$ symbols of the source input. Let the source sequences be denoted by $X(1:k,1:m), Y(1:k,1:m), Z(1:k,1:m)$, where we have broken the source vectors into $m$ blocks of length $k$. In this notation $X(i,j)$ is the $i$th element of the $j$th block, and $X(1:k,j), j\in [1,m]$ is the $j$th block.
\\\textbf{Step 1:} Encoder 1 uses the code $C_k$ to transmit each of the blocks $X(1:k,i), i\in [1,m]$ to the decoder. The second encoder finds the output of the code $C_k$ when $Y(1:k,i)$ is fed to the code, and transmits the output vector on $X_{22}(1:n,i)$. The encoder uses an interleaving method similar to the one in \cite{FinLen} to transmit $Z$. For the sequence $Z(1:k,1:m)$, it finds the output of $C''_{km}$ for this input and transmits it on $X_{21}(1:n,1:m)$.
\\\textbf{Step 2:} The first encoder transmits $B(X(1:k,1:m))$ to the decoder losslessly using $C'_{kmR'}$.
\\\textbf{Decoding:} In the first step, the first decoder reconstructs $X(1:k,1:m)$ with average distortion at most $\gamma+d$. In the second step, using the bin number $B(X(1:k,1:m))$ it can losslessly reconstruct the source, since $C'_{kmR'}$ is a good channel code. Decoder 2 also recovers $Z(1:k,1:m)$ losslessly using $Y_2(1:k,1:m)$ since $C''_{km}$ is a good channel code. 

The conditions for successful transmission is given as follows:
\begin{align*}
 n(1-h_b(\delta))-&k\left(H(X)-h_b(\alpha_X\ast d)\right)
 \\&\geq \sqrt{nV+k\mathcal{V}(d)}Q^{-1}(\gamma)+O({log(n)}),\\
 &n(1-\epsilon)^klog(q)\geq kH(Z).
\end{align*}
Simplifying these conditions by replacing $k=n\left(1-\frac{h_b(d+\gamma)}{1-h_b(\delta)}\right)^{-1}$ proves the proposition.
\end{IEEEproof}
 
 The bound provided in Proposition \ref{prop:ICcorrach} is not calculable without the exact characterization of the $O(\frac{log(n)}{n})$ term. However, we use this bound to prove the sub-optimality of SLCSs. First, we argue that the transmissible region is `continuous' as a function of $\epsilon$. Note that for $\epsilon=0$, sources with parameters $(H(X),H(Z))=(1-h_b(\delta), \log{q})$ are transmissible. The region in Proposition \ref{prop:ICcorrach} is continuous in the sense that as $\epsilon$ approaches 0, the pairs $(H(X),H(Z))$ in the neighborhood of $(1-h_b(\delta), \log{q})$ satisfy the bounds given in the proposition (i.e. the corresponding sources are transmissible).
\begin{Proposition}
For all $\lambda>0$, there exist $\epsilon_0, \gamma_0>0$, and $n_0\in \mathbb{N}$ such that: 
 \begin{align*}
 \forall \epsilon<\epsilon_0: (1-h_b(\delta)-\lambda, \log{q}-\lambda)\in S(n_0,\epsilon, \gamma_0).
\end{align*}
\end{Proposition}
For an arbitrary encoding scheme operating on blocks of length $n$, let the encoding functions be as follows: $X_1^n=\underline{e}_1(X^n)$, and $(X_{21}^n, X_{22}^n)=(\underline{e}_{21}(Y^n,Z^n), \underline{e}_{22}(Y^n,Z^n))$. The following lemma gives an outer bound on $H(Z)$ as a function of the correlation between the outputs of $\underline{e}_1$ and $\underline{e}_{21}$.
\begin{Lemma}
 For a coding scheme with encoding functions $\underline{e}_1(X^n), \underline{e}_{21}(Y^n,Z^n), \underline{e}_{22}(Y^n,Z^n)$, the following holds:
 \begin{align}
 \label{eq:Hz}
 H(Z)\leq \frac{1}{n}\sum_{i=1}^n P(e_{1,i}(X^n)=e_{22,i}(Y^n,Z^n))+1.
\end{align}
\end{Lemma}
\begin{IEEEproof}
 Since $Z^n$ is reconstructed losslessly at the decoder, by Fano's inequality the following holds:
 \begin{align*}
 &H(Z^n)\approx I(Y_2^n;Z^n)\stackrel{(a)}{=} I(E^n,Y_2^n;Z^n)= I(E^n;Z^n)+I(Y_2^n;Z^n|E^n)\\
 &\stackrel{(b)}\leq H(E^n)+\sum_{i=1}^n P(e_{1,i}(X^n)=e_{22,i}(Y^n,Z^n))\log{q}
 \\&\stackrel{(c)}{\leq} n+\sum_{i=1}^n P(e_{1,i}(X^n)=e_{22,i}(Y^n,Z^n))\log{q},
\end{align*}
where in (a) we have defined $E^n$ as the indicator function of the event that $Y_2=e$, in (b) we have used Equation \ref{eq:y2}
and in (c) we have used the fact that $E^n$ is binary.
 \end{IEEEproof}
Using Theorem \ref{th:main}, we show that if the encoding functions are generated using SLCSs, $P(e_{1,i}(X^n)=e_{22,i}(Y^n,Z^n))$ is discontinuous in $\epsilon$. The next proposition shows that SLCSs are sub-optimal:
\begin{Proposition}
 There exists $\lambda>0$, and $q\in \mathbb{N}$, such that sources with $(H(X),H(Z))=(1-h_b(\delta)-\lambda, \log{q}-\lambda)$ are not transmissible using SLCSs. 
\end{Proposition}
\begin{IEEEproof}
 Let $X_1^n=\underline{E}_{1}(X^n)$, and $X_{22^n}=\underline{E}_{22,z^n}(Y^n), z^n\in \{0,1\}^n$ be the encoding functions used in the two encoders to generate $X_1$ and $X_{22}$. If $H(Z)\approx \log(q)$, from \eqref{eq:Hz}, we must have $P(E_{1,j}(X^n)=E_{22,z^n,j}(Y^n))\approx 1 $ for almost all of the indices $j\in[1,n]$. From Theorem \ref{th:main}, this requires $\mathbf{P}_{j,\mathbf{i}_{j}}\approx1$, which requires uncoded transmission (i.e. $X_1^n=\underline{E}(X^n)\approx X^n$). However, uncoded transmission contradicts the lossless reconstruction of the source at the first decoder. 
  \end{IEEEproof}
The proof is not restricted to any particular scheme, rather it shows that any SLCS would have sub-optimal performance.

\section{Conclusion}\label{sec:conc}
We characterized a set of properties which are shared between the SLCSs used in the literature. We showed that schemes which have these properties produce encoding functions which are inefficient in preserving correlation. We derived a probabilistic upper-bound on the correlation between the outputs of random encoders generated using SLCSs. We showed that the correlation between the outputs of such encoders is discontinuous with respect to the input distribution. We used this discontinuity to show that all SLCSs are sub-optimal in a specific multi-terminal communications problem involving the transmission of correlated source over the interference channel. 
\section*{Acknowledgements}
The authors are grateful to Dr. D. Neuhoff for the stimulating discussions and helpful comments. 
\appendix
\subsection{Proof of Theorem \ref{th:sec4}}
\begin{IEEEproof}
The following proposition shows that the probability $P_{\mathscr{S}}(\sum_{{\mathbf{i}}:N_{\mathbf{i}}\leq m, \mathbf{i}\neq \mathbf{i}_k} \mathbf{P}_{k,{\mathbf{i}}}\geq \gamma)$ is independent of the index $k$. This is due to property 1) in Definition \ref{def:SLCS} of SLCS's.
 \begin{Proposition}
 $P(\sum_{{\mathbf{i}}:N_{\mathbf{i}}\leq m, \mathbf{i}\neq 00\cdots01} \mathbf{P}_{k,{\mathbf{i}}}\geq \gamma)$ is constant in $k$. 
\end{Proposition}
\begin{IEEEproof}

Fix $k,k'\in \mathbb{N}$. Define the permutation ${\pi_{k\to k'}}\in S_n$ as the permutation which switches the $k$th and $k'$th elements and fixes all other elements. Also, let $\mathscr{E}$ be the set of all mappings $e:\{0,1\}^n\to \{0,1\}^n$. 
\begin{align*}
  &P_{\mathscr{S}}(\sum_{{\mathbf{i}}:N_{\mathbf{i}}\leq m, \mathbf{i}\neq \mathbf{i}_k} \mathbf{P}_{k,{\mathbf{i}}}>\gamma)
  =\sum_{\underline{e}\in \mathscr{E}}
 P_{\mathscr{S}}(\underline{e})\mathbbm{1}(\sum_{{\mathbf{i}}:N_{\mathbf{i}}\leq m, \mathbf{i}\neq  \mathbf{i}_k|\underline{e}} \mathbf{P}_{k,{\mathbf{i}}}>\gamma)
\\&\stackrel{(a)}{=}\sum_{\underline{e}\in \mathscr{E}}
 P_{\mathscr{S}}(\underline{e}_{\pi_{k\to k'}})\mathbbm{1}(\sum_{{\mathbf{i}}:N_{\mathbf{i}}\leq m, \mathbf{i}\neq \mathbf{i}_k} \mathbf{P}_{k,{\mathbf{i}}}>\gamma|\underline{e})\\
  &\stackrel{(b)}{=}\sum_{\underline{g}\in \mathscr{E}}
 P_{\mathscr{S}}(\underline{g})\mathbbm{1}(\sum_{{\mathbf{i}}:N_{\mathbf{i}}\leq m, \mathbf{i}\neq \mathbf{i}_k} \mathbf{P}_{{\pi_{k\to k'}}{k},{\pi_{k\to k'}}{\mathbf{i}}}>\gamma|\underline{g})
 \\& =\sum_{\underline{g}\in \mathscr{E}}
 P_{\mathscr{S}}(\underline{g})\mathbbm{1}(\sum_{{\underline{l}}:N_{\underline{l}}\leq m, \underline{l}\neq{\pi_{k\to k'}}\mathbf{i}_k} \mathbf{P}_{k',{\underline{l}}}>\gamma|\underline{g})
  =P_{\mathscr{S}}(\sum_{{\mathbf{i}}:N_{\mathbf{i}}\leq m, \mathbf{i}\neq \mathbf{i}_{k'}} \mathbf{P}_{k',{\mathbf{i}}}>\gamma).
  \end{align*}
Where in (a) we have used property (2) 3) in Definition \ref{def:SLCS}, and in (b) we have defined $\underline{g}\triangleq \underline{e}_{{\pi_{k\to k'}}}$ and used ${\pi^2_{k\to k'}}=1$.
  
\end{IEEEproof}

Using the previous proposition, it is enough to show the theorem holds for $k=1$. For ease of notation we drop the subscript k for the rest of the proof and denote $\mathbf{P}_{1,\mathbf{i}}$ by $\mathbf{P}_{\mathbf{i}}$. 
By the Markov inequality, we have the following:
\begin{align}
 P_{\mathscr{S}}(\sum_{{\mathbf{i}}:N_{\mathbf{i}}\leq m, \mathbf{i}\neq \mathbf{i}_1} \mathbf{P}_{{\mathbf{i}}}\geq \gamma)\leq \frac{\sum_{{\mathbf{i}}:N_{\mathbf{i}}\leq m, \mathbf{i}\neq \mathbf{i}_1}\mathbb{E}_{\mathscr{S}}(\mathbf{P}_{\mathbf{i}})}{\gamma}.
 \label{eq:Markov}
\end{align}
So, we need to show that $\sum_{{\mathbf{i}}:N_{\mathbf{i}}\leq m, \mathbf{i}\neq \mathbf{i}_1}\mathbb{E}_{\mathscr{S}}(\mathbf{P}_{\mathbf{i}})$ goes to 0 for all fixed $m$. We first prove the following claim.

\begin{Claim}
Fix $\mathbf{i}$, such that $N_\mathbf{i}\leq m$, the following holds:
\begin{align*}
 \mathbb{E}_{\tilde{E}, X_{\mathbf{i}}}(\mathbb{E}^2_{X^n|X_{\mathbf{i}}}(\tilde{E}|X_{\mathbf{i}}))= \mathbb{E}_{ X_{\mathbf{i}}}(\mathbb{E}^2_{\tilde{E}, X^n|X_{\mathbf{i}}}(\tilde{E}|X_{\mathbf{i}}))+O(e^{-n\delta_X}).
 \end{align*}

\label{claim:expo}
\end{Claim}
\begin{IEEEproof}
\begin{align*}
& \mathbb{E}_{\tilde{E}, X_{\mathbf{i}}}(\mathbb{E}^2_{X^n|X_{\mathbf{i}}}(\tilde{E}|X_{\mathbf{i}}))=\sum_{x_{\mathbf{i}}, \tilde{e}}P(x_{\mathbf{i}})P( \tilde{e})(\sum_{x_{\sim \mathbf{i}}}P(x_{\sim \mathbf{i}})\tilde{e}(x^n))^2\\
 &=\sum_{x_{\mathbf{i}}, \tilde{e}}P(x_{\mathbf{i}})P( \tilde{e})\sum_{x_{\sim \mathbf{i}}}\sum_{y^n: y_{\mathbf{i}}=x_{\mathbf{i}}}P(x_{\sim \mathbf{i}})P(y_{\sim \mathbf{i}})\tilde{e}(x^n)\tilde{e}(y^n)\\
 &=\sum_{x^n}P(x^n)\sum_{y^n:y_{\mathbf{i}}=x_{\mathbf{i}}}P(y_{\sim \mathbf{i}})\mathbb{E}_{\tilde{E}}(\tilde{E}(x^n)\tilde{E}(y^n))\\
 &=\sum_{x^n}P(x^n)\sum_{y^n:y_{\mathbf{i}}=x_{\mathbf{i}}, y^n \in B_n(x^n)}P(y_{\sim \mathbf{i}})\mathbb{E}_{\tilde{E}}(\tilde{E}(x^n)\tilde{E}(y^n))+
 \\& 
 \sum_{x^n}P(x^n)\sum_{y^n:y_{\mathbf{i}}=x_{\mathbf{i}}, y^n\notin B_n(x^n)}P(y_{\sim \mathbf{i}})\mathbb{E}_{\tilde{E}}(\tilde{E}(x^n)\tilde{E}(y^n))\\
 &\stackrel{(a)}{\leq} \sum_{x^n}P(x^n)\sum_{y^n:y_{\mathbf{i}}=x_{\mathbf{i}}, y^n \in B_n(x^n)}P(y_{\sim \mathbf{i}})+
 \\&
 \sum_{x^n}P(x^n)\sum_{y^n:y_{\mathbf{i}}=x_{\mathbf{i}}, y^n\notin B_n(x^n)}P(y_{\sim \mathbf{i}})\mathbb{E}_{\tilde{E}}(\tilde{E}(x^n)\tilde{E}(y^n))\\
 &=P(Y^n\in B_{n}(X^n)|Y_{\mathbf{i}}=X_{\mathbf{i}})+ 
 \\&
 \sum_{x^n}P(x^n)\sum_{y^n:y_{\mathbf{i}}=x_{\mathbf{i}}, y^n\notin B_n(x^n)}P(y_{\sim \mathbf{i}})\mathbb{E}_{\tilde{E}}(\tilde{E}(x^n)\tilde{E}(y^n))\\
 &\stackrel{(b)}{=}O(e^{-n\delta_X})+\sum_{x^n}P(x^n)\sum_{y^n:y_{\mathbf{i}}=x_{\mathbf{i}}, y^n\notin B_n(x^n)}P(y_{\sim \mathbf{i}})\mathbb{E}_{\tilde{E}}(\tilde{E}(x^n))\mathbb{E}_{\tilde{E}}(\tilde{E}(y^n))\\
 &\leq O(e^{-n\delta_X})+P(Y^n\in B_{n}(X^n)|Y_{\mathbf{i}}=X_{\mathbf{i}})+
 \\&
\sum_{x_{\mathbf{i}}}P(x_{\mathbf{i}})\sum_{x_{\sim \mathbf{i}}}\sum_{y^n:y_{\mathbf{i}}=x_{\mathbf{i}}}P(x_{\sim \mathbf{i}})P(y_{\sim \mathbf{i}})\mathbb{E}_{\tilde{E}}(\tilde{E}(x^n))\mathbb{E}_{\tilde{E}}(\tilde{E}(y^n))\\
&=O(e^{-n\delta_X})+\mathbb{E}_{X_{\mathbf{i}}}(\mathbb{E}^2_{\tilde{E}, {X^n|X_{\mathbf{i}}}}(\tilde{E}|X_{\mathbf{i}})).
\end{align*}
In (a) we use the fact that $\tilde{E}\leq 1$ by definition, in (b) follows from property 1) in Definition \ref{def:SLCS}. Define $\bar{E}_{\mathbf{i}}= \mathbb{E}_{\tilde{E}}(\tilde{E}_{\mathbf{i}})= \mathbb{E}_{\tilde{E}|X_{\mathbf{i}}}(\tilde{E}|X_{\mathbf{i}})-\sum_{\mathbf{j}<\mathbf{i}}\bar{E}_{\mathbf{j}}$, and also define $\bar{P}_{\mathbf{i}}\triangleq Var(\bar{E}_{\mathbf{i}}) $.
\end{IEEEproof}
 Using the above claim we have:
\begin{align}
 \nonumber&P_{\mathscr{S}}(\sum_{{\mathbf{i}}:N_{\mathbf{i}}\leq m, \mathbf{i}\neq \mathbf{i}_1} \mathbf{P}_{{\mathbf{i}}}\geq \gamma)
 \leq \frac{\sum_{{\mathbf{i}}:N_{\mathbf{i}}\leq m, \mathbf{i}\neq \mathbf{i}_1}\mathbb{E}_{\mathscr{S}}(\mathbf{P}_{\mathbf{i}})}{\gamma}
 \\&\leq \frac{2^mO(e^{-n\delta_X})+     \sum_{{\mathbf{i}}:N_{\mathbf{i}}\leq m}\mathbb{E}_{\mathscr{S}}(\bar{P}_{\mathbf{i}})-\mathbb{E}_{\mathscr{S}} (\bar{P}_{\mathbf{i}_1})}{\gamma}.
 \label{eq:follow}
\end{align}

Using the arguments from the proof of Proposition \ref{prop:belong2}, we can see that the properties stated in that Proposition hold for $\bar{E}_{\mathbf{i}}$ as well.  Using the same results as in Lemma \ref{Lem:power}, we have that $\sum_{\mathbf{i}\in\{0,1\}^n}\bar{P}_{\mathbf{i}}=\bar{P}_{\mathbf{\underline{1}}}$. Following the calculations in \eqref{eq:follow}:
\begin{align*}
 P_{\mathscr{S}}(\sum_{{\mathbf{i}}:N_{\mathbf{i}}\leq m, \mathbf{i}\neq \mathbf{i}_1} \mathbf{P}_{{\mathbf{i}}}&\geq \gamma)
 \leq \frac{2^mO(e^{-n\delta_X})+     \sum_{{\mathbf{i}}:N_{\mathbf{i}}\leq m}\mathbb{E}_{\mathscr{S}}(\bar{P}_{\mathbf{i}})-\mathbb{E}_{\mathscr{S}} (\bar{P}_{\mathbf{i}_1})}{\gamma}\\
& 
 \leq \frac{2^mO(e^{-n\delta_X})+     \sum_{{\mathbf{i}}\in \{0,1\}^n}\mathbb{E}_{\mathscr{S}}(\bar{P}_{\mathbf{i}}) -\mathbb{E}_{\mathscr{S}} (\bar{P}_{\mathbf{i}_1})}{\gamma} \\
 &=
 \frac{2^mO(e^{-n\delta_X})+    \mathbb{E}_{\mathscr{S}}( \sum_{{\mathbf{i}}\in \{0,1\}^n}\bar{P}_{\mathbf{i}}) -\mathbb{E}_{\mathscr{S}} (\bar{P}_{\mathbf{i}_1})}{\gamma} \\
 &= \frac{2^mO(e^{-n\delta_X})+    \mathbb{E}_{X^n}\left(\mathbb{E}^2_{\tilde{E}|X^n}(\tilde{E}(X^n)|X^n)\right) -\mathbb{E}_{\mathscr{S}} (\bar{P}_{\mathbf{i}_1})}{\gamma}\\
 &\leq \frac{2^mO(e^{-n\delta_X})+    \mathbb{E}_{\mathscr{S}}(\bar{P}_{\mathbf{i}_1})+O(\epsilon) -\mathbb{E}_{\mathscr{S}} (\bar{P}_{\mathbf{i}_1})}{\gamma}\\
 &= \frac{2^mO(e^{-n\delta_X})+O(\epsilon) }{\gamma}
\end{align*}
Where in the last inequality we have used the second property in Definition \ref{def:SLCS}. The last line goes to 0 as $n\to\infty$. This completes the proof.

\end{IEEEproof}
\subsection{Proof of Theorem \ref{th:main}}

\begin{IEEEproof}
  From Theorem \ref{th:sec3}, we have:
  \begin{align*}
     \mathbf{P}^{\frac{1}{2}}\mathbf{Q}^{\frac{1}{2}} -2\sum_{\mathbf{i}}C_{\mathbf{i}}\mathbf{P}^{\frac{1}{2}}_{\mathbf{i}}\mathbf{Q}^{\frac{1}{2}}_{\mathbf{i}}\leq  P(E(X^n)\neq F(Y^n)).
\end{align*}
From Theorem \ref{th:sec4} we have:
\begin{align}
& \forall  m\in \mathbb{N}, \gamma>0,
\\& P_{\mathscr{S}}(\sum_{{\mathbf{i}}:N_{\mathbf{i}}\leq m, \mathbf{i}\neq \mathbf{i}_1} \mathbf{P}_{{\mathbf{i}}}< \gamma)\to 1, \qquad P_{\mathscr{S}}(\sum_{{\mathbf{i}}:N_{\mathbf{i}}\leq m, \mathbf{i}\neq \mathbf{i}_1} \mathbf{Q}_{{\mathbf{i}}}< \gamma)\to 1.
\label{eq:bound}
\end{align}
Note that:
\begin{align}
& \sum_{{\mathbf{i}}:N_{\mathbf{i}}\leq m, \mathbf{i}\neq \mathbf{i}_1} \mathbf{P}_{{\mathbf{i}}}< \gamma ,
 \sum_{{\mathbf{i}}:N_{\mathbf{i}}\leq m, \mathbf{i}\neq \mathbf{i}_1} \mathbf{Q}_{{\mathbf{i}}}< \gamma
 \\&  \Rightarrow
 \sum_{\mathbf{i}}C_{\mathbf{i}}\mathbf{P}^{\frac{1}{2}}_{\mathbf{i}}\mathbf{Q}^{\frac{1}{2}}_{\mathbf{i}} > (1-2\epsilon)(\mathbf{P}_{\mathbf{i}_1}+\gamma)^{\frac{1}{2}}(\mathbf{Q}_{\mathbf{i}_1}+\gamma)^{\frac{1}{2}}+ (1-2\epsilon)^m \mathbf{P}^{\frac{1}{2}}\mathbf{Q}^{\frac{1}{2}},
 \label{eq:bound2}
\end{align}
which converges to $(1-2\epsilon)\mathbf{P}_{\mathbf{i}_1}^{\frac{1}{2}}\mathbf{Q}_{\mathbf{i}_1}^{\frac{1}{2}}+(1-2\epsilon)^m\mathbf{P}^{\frac{1}{2}}\mathbf{Q}^{\frac{1}{2}}$ as $\gamma\to 0$. Also $C_{\mathbf{i}}$ is decreasing in $N_{\mathbf{i}}$ and goes to 0 as $N_{\mathbf{i}}\to \infty$. Choose $\gamma$ small enough and $m$ large enough such that $(1-2\epsilon)(\mathbf{P}_{\mathbf{i}_1}+\gamma)^{\frac{1}{2}}(\mathbf{Q}_{\mathbf{i}_1}+\gamma)^{\frac{1}{2}}+(1-2\epsilon)^m \mathbf{P}^{\frac{1}{2}}\mathbf{Q}^{\frac{1}{2}}-(1-2\epsilon)\mathbf{P}_{\mathbf{i}_1}^{\frac{1}{2}}\mathbf{Q}_{\mathbf{i}_1}^{\frac{1}{2}}< \delta $. Then Equations \eqref{eq:bound} and \eqref{eq:bound2} gives 
\[
P_{\mathscr{S}}\left(P_{X^n,Y^n}\left(E(X^n)\neq F(Y^n)\right)<2\mathbf{P}^{\frac{1}{2}}\mathbf{Q}^{\frac{1}{2}}- 2(1-2\epsilon)\mathbf{P}_{\mathbf{i}_1}^{\frac{1}{2}}\mathbf{Q}_{\mathbf{i}_1}^{\frac{1}{2}}-\delta\right)\to 0.\]
This is equivalent to the statement of the theorem.
\end{IEEEproof}

\end{document}